\documentclass[runningheads]{llncs}

\usepackage{cite}
\usepackage{amsmath, amssymb, amsfonts}
\usepackage{algorithm}
    \usepackage[noEnd]{algpseudocodex}
\usepackage{graphicx}
\usepackage{svg}
\usepackage{textcomp}
\usepackage{xcolor}
\usepackage{mathtools}

\newcommand{\equal} {\mathcal{E}}
\newcommand{\add} {\mathrm\Pi_{\operatorname{add}}}
\newcommand{\mul} {\mathrm\Pi_{\operatorname{mul}}}

\newcommand{\sqrtprot} {\mathrm\Pi_{\operatorname{sqrt}}}

\newcommand{\invprot} {\mathrm\Pi_{\operatorname{inv}}}
\newcommand{\ca}{\mathrm\Pi_{\operatorname{\mathcal{E}}}}

\usepackage[pdfborder={0 0 0}]{hyperref}

\makeatletter\renewcommand{\ALG@name}{Alg.}\makeatother

\begin{document}

\title{Confidential Truth Finding with~Multi-Party~Computation (Extended~Version)}

\author{Angelo Saadeh\inst{1,6} \and Pierre Senellart \inst{2,4,5,6,7}
\and Stéphane Bressan \inst{3,6,7}}

\institute{%
    LTCI, Télécom Paris, IP Paris, Palaiseau, France\\
\email{angelo.saadeh@telecom-paris.fr}\\
\and DI~ENS, ENS, PSL University, CNRS, Paris, France\\
\email{pierre@senellart.com}
\and National University of Singapore, Singapore\\ \email{steph@nus.edu.sg}
\and Inria, Paris, France
\and Institut Universitaire de France
\and CNRS@CREATE LTD,
Singapore
\and IPAL, CNRS, Singapore}

\maketitle

\begin{abstract}
Federated knowledge discovery and data mining are challenged to assess the trustworthiness of data originating from autonomous sources while protecting confidentiality and privacy. Truth-finding algorithms help corroborate data from disagreeing sources. For each query it receives, a truth-finding algorithm predicts a truth value of the answer, possibly updating the trustworthiness factor of each source.
Few works, however, address the issues of confidentiality and privacy.
We devise and present a secure secret-sharing-based multi-party computation protocol for pseudo-equality tests that are used in truth-finding algorithms to compute additions depending on a condition.
The protocol guarantees confidentiality of the data and privacy of the sources.
We also present variants of truth-finding algorithms that would make the computation faster when executed using secure multi-party computation.
We empirically evaluate the performance of the proposed protocol on two state-of-the-art truth-finding algorithms, Cosine, and 3-Estimates, and compare them with that of the baseline plain algorithms. The results confirm that the secret-sharing-based secure multi-party algorithms are as accurate as the corresponding baselines but for proposed numerical approximations that significantly reduce the efficiency loss incurred.

\keywords{truth finding  \and secure multi-party computation \and secret-sharing \and uncertain data \and privacy.}
\end{abstract}

\section{Introduction}

Federated knowledge discovery and data
mining~\cite{federated-learning-introduction,federated-learning-survey}
are challenged to assess the trustworthiness of data originating from
autonomous sources while protecting confidentiality. Truth-finding
algorithms~\cite{TF-survey} help corroborate data from disagreeing
sources. For each query it receives, a truth-finding algorithm predicts a
truth value of the answer, possibly updating the trustworthiness factor
of each source. Few works, however, address the issues of confidentiality
and privacy. We consider the design and implementation of truth-finding
algorithms that protect the confidentiality of sources' data.

For example, a creditor may wish to determine whether loan applicants are
creditworthy. The creditor would want to base her decision on the different, possibly disagreeing, evaluations of the applicants by several financial institutions. However, financial institutions only agree to contribute their respective evaluations to the decision provided it is neither revealed to the creditor, to the other financial institutions, nor to third parties. For such a purpose, we turn to secret-sharing-based secure multi-party computation~\cite{mpcbook-damgard}, or simply secure multi-party computation (MPC).

We devise and present a secure multi-party pseudo-equality protocol that securely computes additions depending on a condition -- we call them conditioned additions -- for truth-finding algorithms. In particular, we present a secure equality test alternative that uses a polynomial evaluation to reduce the number of communication; this is used for conditioned additions, an operation that is an essential building block of many truth-finding algorithms. The protocol guarantees the confidentiality of the data.

We also devise several variants of privacy-preserving truth-finding algorithms; ones that implement the truth-finding algorithms without changes, and others with modifications that aim to make the computation more efficient.

The secure multi-party protocols are then implemented with two servers.
In the running example, the two servers can be operated by two
non-colluding entities such as independent third parties, the creditor,
or the financial institutions. We empirically evaluate the performance of
the proposed protocol on two state-of-the-art truth-finding algorithms,
Cosine~\cite[Algorithm 1]{cosine} and 3-Estimates~\cite[Algorithm
4]{cosine} (see also \cite{li2013truth,berti2015data} for further
experiments on these algorithms), and compare them with that of
the non-secure baseline algorithms. The results confirm that the secure
multi-party algorithms are as accurate as the corresponding baselines
except for proposed modifications to reduce the efficiency loss incurred.

\medskip

Set \(n\in\mathbb{N^*}\), and let \(\mathcal{V}\) be a set of \(n\)
sources. The client would like to label \(k\) queries (or facts)
\(\{f^1,...,f^k\}\). A truth-finding algorithm outputs a truth value for
a query when different data sources (or sources) provide disagreeing
information on it. Concretely, the truth-finding algorithm takes
\(v^1,...,v^n\) as input with \(v^i\in\{-1,0,1\}^k\), and outputs
estimated truth values in \([-1,1]^k\subset\mathbb{R}^k\) or \([0,1]^k\subset\mathbb{R}^k\) depending on the truth-finding algorithm.

Truth-finding (or truth discovery) algorithms \cite{TF-survey} are usually run by the client in order to know the truth value of a given query when the sources give disagreeing answers. That is, for each of the client's queries, each source in \(\mathcal{V}\) delivers an answer \(v^i\) such that an output of \(1\) corresponds to a positive answer, \(-1\) to a negative one, and \(0\) if the source does not wish to classify the data point.
Cosine and 3-Estimates \cite{cosine} are two truth-finding algorithms
that given a number of queries \(k\), output a truth value in the range
\([-1,1]^k\subset\mathbb{R}\) and a trust coefficient in each of the
sources, or sources. In addition, 3-Estimates computes an estimate of the difficulty of each query.

The goal of this work is to execute truth-finding algorithms that protect sources' data using secure multi-party computation (MPC)~\cite{BGWmodel,mpcbook-damgard}. More generally, given a function \(F\) and a set of private inputs \(x^1,...,x^m\) respectively owned by \(P_1,..., P_m\), MPC is a cryptographic approach that makes it possible to compute the output of the function \(F(x^1,..,x^m)\) without resorting to a third party that would compute the function \(F\) and would send the result back. MPC will be used to implement the Cosine and 3-Estimates algorithms without having any source disclose their answer.

\section{Background}
\paragraph{Truth finding.}
Truth finding \cite{truth-finding-2007,cosine,li2013truth,berti2015data}
is an effective tool used to handle uncertain data. More specifically,
when a dataset is missing some information and the dataset owner does not
have access to this information, they can ask sources questions (or
queries) in order to complete the dataset. Yet again, the sources may not
be completely sure of the answer they are delivering. Truth-finding
algorithms rely on the correlation between the answers of all
sources. Furthermore, the client does not have any information about how
the sources get their information, i.e., how they construct their model, and how they take their decisions. In other words, in real applications, the process is completely unsupervised and this is why truth-finding algorithms are used.

Consider a set of facts \(\{f^1,...,f^k\}\) and a set of \(n\)
sources. Each source can map the fact \(f^i\) to \(\{-1,0,1\}\), and the
image of the facts computed by a source represents the source's view of
the facts. A negative value represents a false fact, a
positive value represents a true fact, and a null value
is an undetermined fact. We set \(v^1,...,v^n\) to be the sources' views
of the facts, more precisely, \(v^i\in\{-1,0,1\}^k\) is the \(i\){th} source's view of the facts \(\{f^1,...,f^k\}\).

In this paper, we consider two existing truth-finding algorithms for which
we apply MPC, though we stress that our
overall approach only marginally depends on the specific algorithm used.

The idea of Cosine, based on cosine similarity in information
retrieval methods \cite{information-retrieval} and precisely described in
\cite[Algorithm 1]{cosine}, is to iteratively compute a truth value for
each fact given the views of all the sources. With each iteration, the
algorithm also updates a trustworthiness factor for each source. In the
end, the algorithm returns one truth value for each fact and one trust
factor for each source. The truth value and trust factor are initialized
and then updated at each iteration as follows: the truth value is computed as the sum of answers compounded with the trust factor of each source, and then the trust factor is computed by normalizing the number of answers that each source got right.

On the other hand, 3-Estimates \cite[Algorithm 3]{cosine} takes a third factor into
consideration: the difficulty of the query. The algorithm outputs
a truth value and trust factor like Cosine, but also outputs a difficulty
factor for each fact (or query). More formally, for a query \(f^j\) if we
set \(\delta^j\) to be the probability of the query \(f^j\) being
difficult and \(\theta^i\) the probability of the $i${th} source (that
didn't answer~\(0\)) being not trustworthy, then the algorithm estimates
the truth value as:
\begin{equation*}\left\{
\begin{aligned}
\Pr(\mbox{source $i$ is wrong on fact $j$}) & \coloneqq \delta^j\theta^i  \\
\Pr(\mbox{source $i$ is right on fact $j$}) & \coloneqq 1 - \delta^j\theta^i
\end{aligned}\right.
\end{equation*}

If \(\lambda^j\) is the number of sources responding to query
\(f^j\), and \(v^{ij}\) the answer of the \(i\){th} source to~\(f^j\),  then the probability of \(f^j\) to be true is given by:
\begin{align*}
    \textstyle\lambda^j \Pr(f^j \mbox{ is true}) &\textstyle =
    \sum_{i,v^{ij}=1}\Pr(\mbox{\(i\) is right on $j$}) +
    \sum_{i,v^{ij}=-1}\Pr(\mbox{\(i\) is wrong on $j$}) \\
    &\textstyle = \sum_{i,v^{ij}=1} 1 - \delta^j\theta^i + \sum_{i,v^{ij}=-1} \delta^j\theta^i
\end{align*}
A similar equation is used to update the difficulty of each query and the trust in each source on each iteration.

\paragraph{Secure multi-party computation.}
Secure multi-party computation (MPC) \cite{BGWmodel,mpcbook-damgard} allows a set of \(m\) players to compute a
function on their private inputs without revealing them to a third party.
The solution proposed by this paper only uses two parties, though this
number could be increased at the cost of lower efficiency. To this end,
we present background on MPC in the specific
case where \(m=2\). Increasing the number of servers to some
arbitrary~\(m\) would tolerate \(m-1\) players colluding with each other in the passive security model. However, in a real-world application, the two servers could be chosen in a way that they have no interest in colluding, for example, one server could be the client, and the second server could be a representative of the sources. Therefore, it is sufficient to consider only two servers and limit the number of communication to a minimum.

Let \(\mathbb{Z}_{2^q} \coloneqq \mathbb{Z}/2^q\mathbb{Z}\).
Suppose each of \(P_1\) and \(P_2\) has a secret \(x^1\) and \(x^2\) both in \(\mathbb{Z}_{2^q}\). Their goal is to compute \(y = F(x^1, x^2)\) where \(F\) is a public function without revealing their respective inputs.
The first step is having each player \(P_i\) mask their secret \(x^i\)
with a random ring element \(x^i_j\), send the mask \(x^i_j\) to the
other player (\(P_j\)) and keep the masked value \(x^i_i= x^i-x^i_j\). Of
course \(x^i = x^i_i+x^i_j\) in the ring; this is called the additive
secret-sharing scheme \cite{MPSPDZ}. The elements with subscripts
correspond to ring elements that seem random but whose sum is equal to the secret; they are called additive shares.

The players then evaluate the arithmetic circuit of \(F\) such that on each addition node a protocol \(\add\) is used and on each multiplication node the protocol \(\mul\) \cite{beaver1991} is used. After evaluating all the nodes, each player \(P_i\) ends up with a value \(y_i\) such that \(y = y_1+y_2\), so the players reveal their final values and add them together.
If \(P_1\) holds \(a_1,b_1\) and \(P_2\) holds \(a_2,b_2\) such that
\(a=a_1+a_2\) and \(b=b_1+b_2\) then for \(i\in\{1,2\}\) \(\add(a_i,b_i)\) allows
\(P_i\) to hold \(c_i\) such that \(a+b =
c_1+c_2\) without learning \(a\) or \(b\). In addition, for \(i\in\{1,2\}\)
\(\mul(a_i,b_i)\) allows \(P_i\) to hold \(c_i\) such that \(a\cdot b = c_1+c_2\) without learning \(a\) or
\(b\). This is why by evaluating the arithmetic circuit with MPC protocols node by node the players would obtain \(y_1\) and \(y_2\). The goal is to find the arithmetic circuit that computes \(F\) or best approximates \(F\).

The secret sharing scheme, addition, and multiplication protocols can be
chosen as a way to satisfy the needed security level; in our work we
implement the minimal security measure which is passive security. In
simple words, the players should not deviate from the protocol and should
not learn information about each other's input unless it can be deduced
from the output. The protocols we consider do not resist an active
adversary, i.e., an adversary that corrupts a player and deviates from
the protocol. For more information about the adversary types and formal
security definitions of MPC, see~\cite{mpc-theory-book}.

Furthermore, we use additive secret-sharing-based MPC, and the protocols are computed in a finite ring. But
the inputs and the operations in truth-finding algorithms are in~\(\mathbb{R}\). We therefore map the
real data inputs to the finite ring using fixed-point precision
\cite{fixedpoint-precision} -- which is classically done in MPC.
For simplicity, we refer with the same notation to the real inputs and
their ring mapping. Real addition and multiplication operations are
approximated by other operations, but for simplicity, we also refer to them with the same notation.

\section{Related Work}
Early works in truth-finding algorithms \cite{truth-finding-2007, cosine}
show that majority voting is not the best solution to corroborate data
when different sources provide conflicting information on it. Interestingly, further
studies~\cite{li2013truth,berti2015data} show that no single
truth-finding algorithm performs well in all scenarios and benchmarks, we
just choose Cosine and 3-Estimates as representative examples of such
algorithms.

Since their introduction, cryptographic privacy-preserving tools like MPC and homomorphic encryption \cite{HE-survey} have been used for federated tasks.
Current state-of-the-art multi-party computation protocols allow players
to compute functions securely, robustly, and efficiently. Many secure
multi-party frameworks have been developed such as~\cite{ABY3} and some of
them are specific for machine learning tasks \cite{crypten,pysyft}.
Alternatively, homomorphic encryption -- also used for machine learning
tasks \cite{concrete-ml-zama} -- can be used for scenarios where no communications take place during the computations; only one party needs to be doing them.

Concerning applications similar to truth finding, cryptography has been
used for e-voting: for example, MPC in \cite{e-voting-mpc} and
homomorphic encryption in \cite{e-voting-he}. Both tools have even been
combined \cite{private-aggregation-without-secure-channel} in order to
achieve a privacy-preserving aggregation without secure communication
channels. These results are for simple majority voting and do not
consider other truth-finding algorithms. Privacy-preserving truth-finding
algorithms were not common until 2015 with \cite{TF-2015} and afterward with \cite{HE-crowdsensing-17, HE-crowdsensing-18} and other similar works; all these works consider the same specific problem of mobile crowd-sensing systems, and they all use homomorphic encryption and one specific truth-finding algorithm: CRH~\cite{crh}.
The aim of this paper is more general, it proposes MPC protocols and
implementation techniques that can be applied to various truth-finding
algorithms. To our knowledge, MPC has not been used to securely evaluate truth-finding algorithms.

\section{Proposed Approach}
The first task we wish to achieve is private voting, i.e., the client
sends queries to each source, and the source classifies the query. In the
case where the query is a vector of features and the models are logistic
regressions, existing MPC works \cite{ABY3} can keep the query private.
We suppose that the answers are already computed and secret-shared on two
servers \(P_1\) and \(P_2\) using a two-party additive secret sharing. In
other words, \(P_1\) holds \(v^{ij}_1\) and \(P_2\) holds
\(v^{ij}_2\) such that \(v^{ij}=v^{ij}_1+v^{ij}_2\) is the \(i\){th}
source's answer for the query \(f^j\) and is equal to \(-1\), \(0\), or \(1\).

The second step which is the aggregation of the data (the answers) is
computed on the two servers $P_1$ and~$P_2$. The problem is now constructing a secure two-party computation algorithm with additively shared data that implements the truth-finding algorithms using their arithmetic circuits. Once the circuits are evaluated, the two servers (\(P_1\) and \(P_2\)) send their share of the output to the client who reconstructs it by adding the received shares together.

\subsection{MPC Protocols for Truth Finding}
Other than additions and multiplications, the truth-finding
algorithms we implement -- Cosine and 3-Estimates -- use existing
real-number
operations like division, and square root.  We also propose a way to compute conditioned sums by replacing equality tests with degree-two polynomial evaluations. We now explain
how these three
operations can be approximated using arithmetic circuits consisting of additions and
multiplications. Furthermore, multiplications are communication-costly in
MPC, so the aim is to use a low number of
multiplications -- or communications in general.

\paragraph{Division and square root.}
The approximations given below are based on numerical methods. These
approximations are widely used in MPC frameworks for machine learning,
like in \cite{pysyft, crypten}, where the real inverse and the square root
are computed iteratively using the Newton--Raphson method. Suppose \(P_1\)
holds \(a_1\) and \(P_2\) holds \(a_2\) such that \(a=a_1+a_2\). Then we
denote by \(\invprot(a_1,a_2)\) the protocol that allows each player to
hold \(b_1\) and \(b_2\) respectively without learning information about
\(a_1\) and \(a_2\), and where \(b_1+b_2=\sqrt{a}\). Similarly, we denote
by \(\invprot(a_1,a_2)\) the protocol that allows each player to hold
\(c_1\) and \(c_2\) respectively without learning information about
\(a_1\) and \(a_2\), and where \(c_1+c_2=\frac{1}{a}\). These protocols
\(\sqrtprot\), \(\invprot\) are both defined as a succession of additions and
multiplications. They can be modeled into an arithmetic circuit which
could be evaluated using existing MPC protocols for additions and
multiplications like \(\add\) and \(\mul\).
Finally,  note that division is multiplication by an
inverse.
Consequently, the evaluation
of the arithmetic circuit satisfies the same level of security as these
two protocols.

\paragraph{Conditioned additions.}
We propose an alternative for the equality test, which is a degree-two polynomial evaluation.
The truth-finding algorithms we use require conditioned additions. Given two vectors of same size \(t=(t^1,...,t^k) \in \mathbb{R}^k\), \(z=(z^1,...,z^k) \in \{-1,0,1\}^k\), and an element \(\kappa\in\{-1,0,1\}\), we define the following operation:
\(S \coloneqq \sum_{i: z^i = \kappa} t^i.\)
In other words, the \(i\)th element of \(t\), \(t^i\), is added to the
sum only if the \(i\){th} element of \(z\), \(z^i\), is equal to
\(\kappa\). The difficulty is that even though \(\kappa\) is public,
\(z^i\) is private. To achieve this in MPC we start by defining the following function, for \(i\in\{1,...,r\}\):
\[\equal(z^i,\kappa)=\left\{
    \begin{aligned}
         1 & \mbox{ if } z^i = \kappa \\
         0 & \mbox{ if not.}
    \end{aligned}\right.
    \]
A naive way to compute the sum \(S\) is as follows:
\(S = \sum_i \equal(z^i,\kappa)\cdot t^i.\) This way to compute \(S\) requires an equality test which is costly in
MPC. To this end, we propose an alternative that makes good use of the fact that \(z^i,\kappa\in\{-1,0,1\}\). The goal is to express the function \(\equal\) as a polynomial so that it can be computed using the smallest number of additions and multiplications possible. We define and use the following expressions of \(\equal(z^i, \kappa)\).

If \(\kappa=-1\), we compute \(S\) as follows:
\(S=\sum_i \frac{1}{2}((z^i)^2-z^i)\cdot t^i.\)
We have:
\[\frac{1}{2}((z^i)^2-z^i) = \left\{
\begin{array}{cc}
    1 & \mbox{ if } z^i=-1 \\
    0 & \mbox{ if } z^i=0 \\
    0 & \mbox{ if } z^i=1
\end{array}
\right.\]
Hence by multiplying \(\frac{1}{2}((z^i)^2-z^i)\) by \(t^i\), the only
elements considered in the sum are the ones such that \(z^i=-1\). The function \(\frac{1}{2}((z^i)^2-z^i)\) is equal to \(\equal(z^i,-1)\).\\
If \(\kappa=0\) we similarly compute \(S\) as:
\(S=\sum_i (1-(z^i)^2) \cdot t^i.\)
It is also straightforward that the function \(1-(z^i)^2\) is equal to \(\equal(z^i,0)\) because it outputs \(1\) if \(z^i=0\) and \(0\) elsewise.
If \(z=1\), in the same way, \(S\) is computed as:
\(S=\sum_i \frac{1}{2}((z^i)^2+z^i)\cdot t^i.\)
\begin{lemma}[Conditioned additions]
Denote by \(\ca\) the MPC protocol implementing the function \(\mathcal{E}\) using the three previously defined degree-2 polynomials. \(\ca\) does not reveal information about the other's player's share.
\end{lemma}
\begin{proof}
The three conditioned sums defined in this section do not need
comparisons and they are expressed using only additions and
multiplications, so their security level is the same as \(\add\) and
\(\mul\).\qed
\end{proof}

\paragraph{3-Estimates with MPC.} Our MPC implementation of 3-Estimates
is given in Alg.~\ref{alg:3est}.
The protocols presented in the previous section allow us to implement the
truth-finding algorithms with MPC.

\begin{algorithm}
    \caption{\hspace*{1em}3-Estimates algorithm with secure multi-party computation}\label{alg:3est}
\begin{algorithmic}
    \Require The answers $(v^{ij})_{i=1..n}^{j=1..k}$ are secret shared on two servers
    \Ensure The client receives $y, \theta, \delta$
    \For{$i=1..n$} \Comment{Initialization of the untrustworthiness of each source}
        \State $\theta^i \gets 0.4$
    \EndFor

    \For{$j=1..k$} \Comment{Initialization of the difficulty of each query}
        \State $\delta^j \gets 0.1$
    \EndFor
    \For{$j=1..k, i=1..n$} \Comment{Compute equality tests for the conditioned sums}
    \State $\sigma^{ij} \gets \ca(v^{ij}, 1)$
    \State $\tau^{ij} \gets \ca(v^{ij}, -1)$
    \EndFor

    \Repeat
    \For{$j=1..k$} \Comment{Update the truth value of each query}
        \State $\mathit{posViews} \gets \sum_{i=1}^{i=n} \sigma^{ij}\cdot(1-\theta^i\delta^j) $
        \State $\mathit{negViews} \gets \sum_{i=1}^{i=n} \tau^{ij}\cdot(\theta^i\delta^j) $
        \State $\mathit{nbViews} \gets \sum_{i=1}^{i=n} \sigma^{ij}+\tau^{ij} $
        \State $y_j \gets (\mathit{posViews} + \mathit{negViews})\cdot \invprot(\mathit{nbViews}) $
    \EndFor
    \State Normalize \(y\)
    \For{$j=1..k$} \Comment{Update the difficulty score of each query}
        \State $\mathit{posViews} \gets \sum_{i=1}^{i=n} \sigma^{ij}\cdot(1-y^j)\cdot\invprot(\theta^i) $
        \State $\mathit{negViews} \gets \sum_{i=1}^{i=n} \tau^{ij}\cdot y^j\cdot\invprot(\theta^i) $
        \State $\mathit{nbViews} \gets \sum_{i=1}^{i=n} \sigma^{ij}+\tau^{ij} $
        \State $\delta^j \gets (\mathit{posViews} + \mathit{negViews})\cdot \invprot(\mathit{nbViews}) $
    \EndFor
    \State Normalize \(\delta\)
    \For{$i=1..n$} \Comment{Update the untrustworthiness of each source}
        \State $\mathit{posFacts} \gets \sum_{j=1}^{k} \sigma^{ij}\cdot(1-y_j) \cdot \invprot(\delta^j) $
        \State $\mathit{negFacts} \gets \sum_{j=1}^{k} \tau^{ij}\cdot y_j \cdot \invprot(\delta^j) $
        \State $\mathit{nbFacts} \gets \sum_{j=1}^{k} \sigma^{ij}+\tau^{ij} $
        \State $\theta^i \gets (\mathit{posFacts} + \mathit{negFacts})\cdot \invprot(\mathit{nbFacts}) $
    \EndFor
    \State Normalize \(\theta\)
    \Until{convergence}

    \State Servers send the shares of $y,\theta, \delta$ to
    Client.
\end{algorithmic}
\end{algorithm}

\begin{theorem}
Alg.~\ref{alg:3est} ensures that the client learns the truth value and the difficulty of each query as well as the trust factor with passive security.
\end{theorem}
\begin{proof}[sketch]
The answers are additively secret-shared on the servers at the beginning,
giving the servers no information about the sources' answers at this point. Then the entire computation takes place in a secret-shared manner by evaluating an arithmetic circuit with secure addition and multiplication protocols, making the rest of the computation secure in the sense that the servers learn no information about the secrets.
\end{proof}

The detailed MPC Cosine algorithm is analogous to Alg.~\ref{alg:3est}; proof
of security is achieved in the same way.

\subsection{MPC-friendly alternative implementations}\label{sec:alg2}
In this section, we propose changes to
Cosine and 3-Estimates to reduce communication costs, at the cost of
a possibly higher number of errors. We also illustrate in this section the MPC version of 3-Estimates, along with a sketch of a security proof.

\paragraph{Normalization in 3-Estimates.}
In the 3-Estimates algorithm, the truth value, trust factor, and difficulty
score need to be normalized at each step. This could be done using a
secure comparison protocol to securely compute the minimum and the
maximum of each value, and then normalize them as it is done in
\cite{cosine}. Secure comparisons however are very costly in MPC. To
reduce the amount of communication we replace the normalization based on
finding the maximum and minimum by a pre-computed linear transformation
which forces the values to stay between \(0\) and \(1\). Concretely we
apply the function \(h(x)=0.5x+0.25\) to all the values after each
update. We evaluate the impact of this change in the experiments. The
chosen function, \(h\), is not perfect. Indeed, if we have information about the distribution of the parameters, we can pre-compute a linear normalization for every iteration. Using any public pre-computed or pre-defined normalizing function improves the efficiency of the algorithm because it would translate to using multiplication and addition by public constants, which is communication-free.

\paragraph{Efficient alternatives for Cosine.}

In Cosine, the truth value and trust factor can be negative, and protocol
\(\invprot\) can only be applied to positive numbers. Consequently, every
time there is a division by an element \(x\), the inverse protocol is
applied to \(|x|\) and then the result is multiplied by the sign of \(x\)
\cite{crypten}. Computing the sign of \(x\) requires computing a secure
comparison, which is communication-costly. With the aim to reduce the
number of communications, we propose inverting \(x^2\) and multiplying by
\(x\). This technique should give the same result with fewer
communications. However, in the Cosine algorithm, the denominators are a
linear combination of \((\theta^i)^3\) -- trust factors of sources to
the cube -- and since the trust factor is between \(-1\) and \(1\),
\((\theta^i)^3\) could be very small, and squaring it for the sake of a
faster inverse makes it even smaller. To avoid any precision issues, we
implement a version of the algorithm where we replace \((\theta^i)^3\) by
\(\theta^i\) which will have an impact on the truth value. This impact,
however, does not affect the sign of the truth value, it only affects its
amplitude, leaving the rounding (i.e., the final label) unchanged. We evaluate the impact of this change in the experiments. Additionally, replacing \((\theta^i)^3\) by \(\theta^i\) saves multiplications.

\section{Experimental Results}
We evaluate our protocols on two computing servers. We suppose that the
sources have already answered and secret-shared their answers. We use the
ring \(\mathbb{Z}_{2^{60}}\) with \(20\) bits of fixed precision. The two
servers communicate via a local socket network implemented in Python on
an Intel Core i5-9400H CPU (\(\text{2.50\,GHz}\times8\)) and a RAM of
15.4\,GiB. For the sake of the experiment, these communications are not encrypted or authenticated. Note that we do not compare our approach with the approaches cited in the related works, as they are based on homomorphic encryption and it is not comparable with secret-sharing-based multi-party computation which is done in a different setting, i.e. the players have to be online during the computation.

\paragraph{3-Estimates on Hubdub Dataset.}
We implement our solution using the dataset Hubdub
from~\cite{cosine}.\footnote{All datasets used, as well as the source code
    of our implementation, are available
at~\url{https://github.com/angelos25/tf-mpc/}.}
This dataset is
constructed from \(457\) questions from a Web site
where users had to bet on future events. As the questions
had multiple answers, they have been increased to \(830\) questions to
obtain binary questions with answers \(-1,0\) or \(1\). The client sends
the \(830\) queries to be classified by each source, and after the
classification, the sources secret-share them on two servers to evaluate
using MPC the 3-Estimates truth-finding algorithm. At the end of the
evaluation, the results are reconstructed by the client. The results
include the truth value for each query (the label), a difficulty score
for each query, and a trustworthiness factor for each of the \(471\)
sources. In Fig.~\ref{fig:hubdub} we show the difference between the
predictions from the base model and the predictions from the MPC
evaluation. The base model corresponds to the 3-Estimates algorithm
implemented without MPC on the plain data.
\begin{figure}[p]
    \centering
    \includegraphics[width=.7\linewidth]{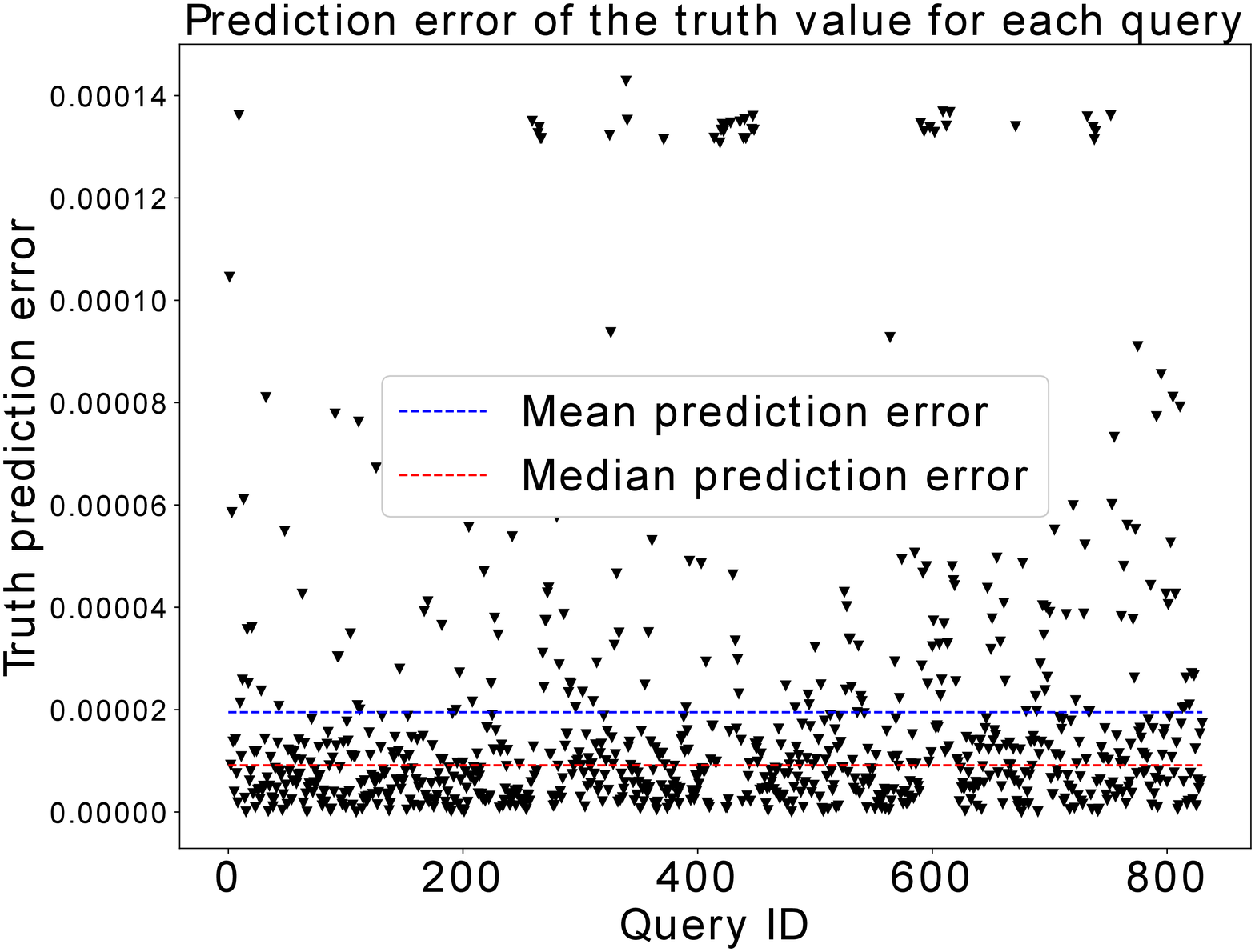}\\[1em]
    \includegraphics[width=.7\linewidth]{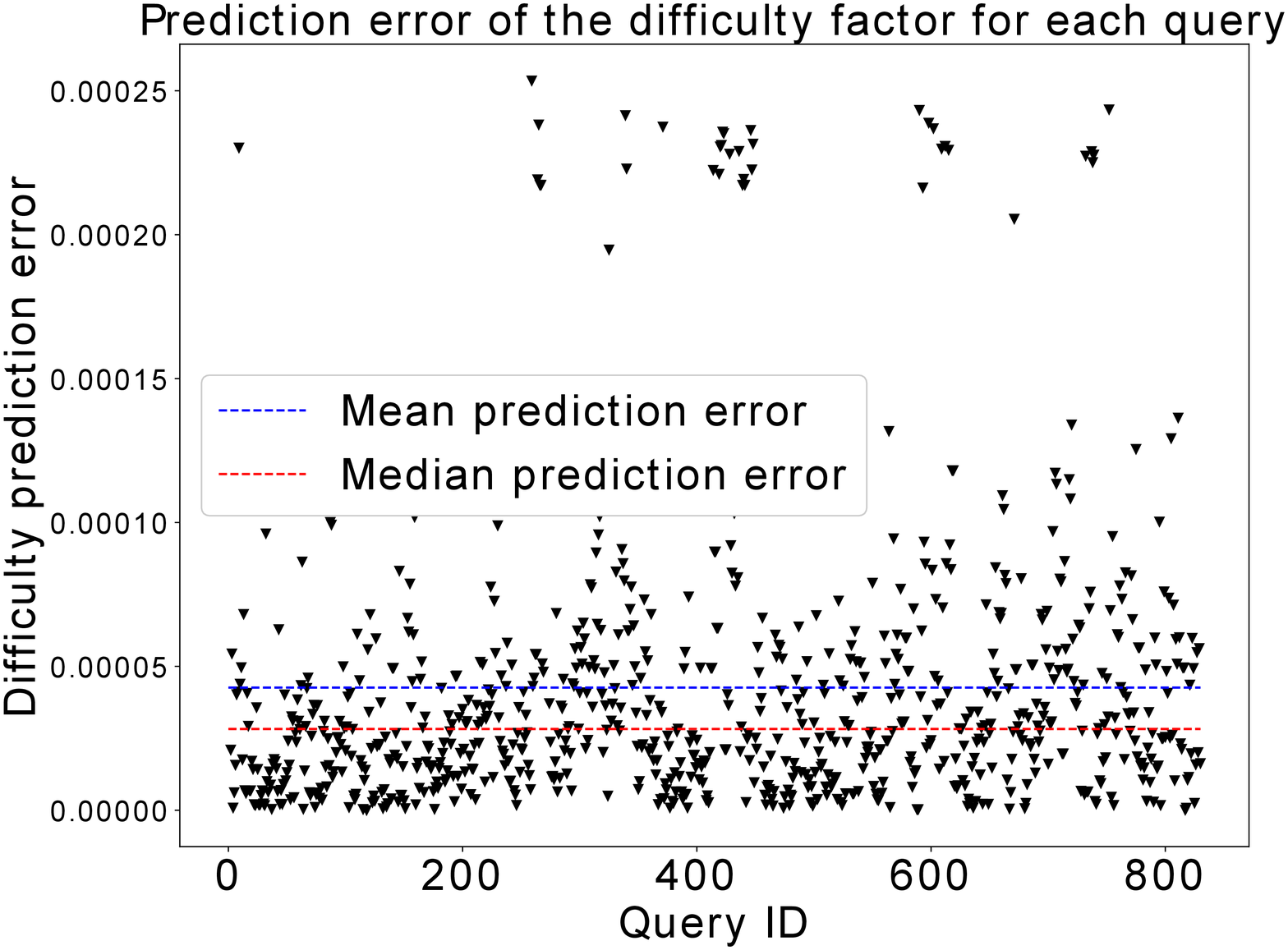}\\[1em]
    \includegraphics[width=.7\linewidth]{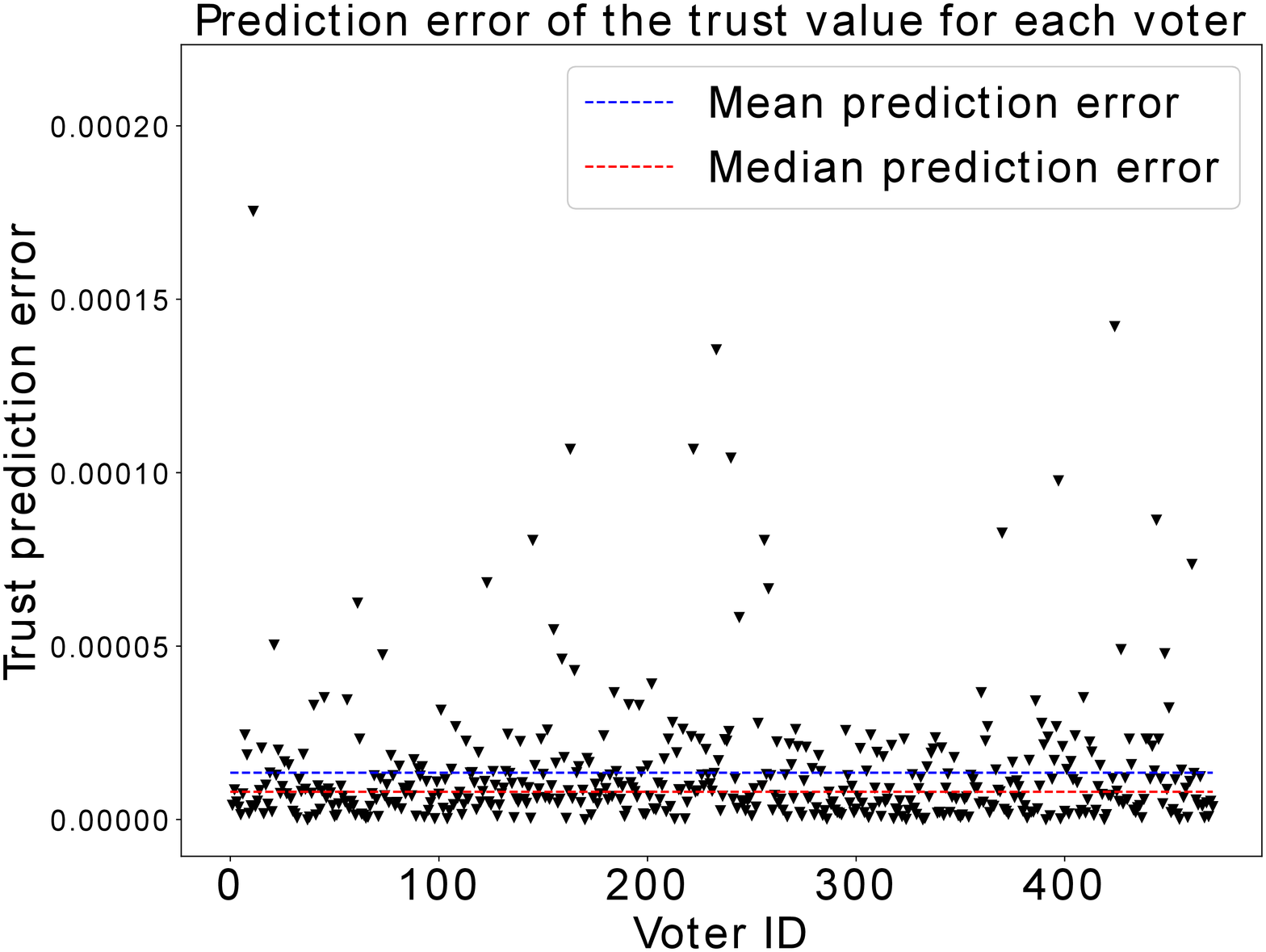}
    \caption{Prediction errors between secure multi-party computation and the base model results with 3-Estimates on Hubdub dataset.}
    \label{fig:hubdub}
\end{figure}
The MPC evaluation contains errors compared to the base model, and these
errors are mostly below \(10^{-4}\). To evaluate the impact of the errors
induced by MPC, we look at label prediction. The MPC method labels all
the questions exactly the same way as the baseline method, so both
methods made the same number of errors, i.e., \(269\) (as shown
in~\cite{cosine}, this is less than majority voting and some other
methods). On average, the
execution of each iteration took \(52.85 s\) wall-clock time, or \(39.58s\)
CPU time. The MPC model is \(2\,000\) times slower than the base model, this is due to the high number of comparisons that should be made to normalize the three factors.

If we use the pre-computed linear function \(h\) presented in Sec.~\ref{sec:alg2} the outputs will be very different of course because of
the aforementioned reasons, but wall-clock time of each iteration is
reduced to \(0.58s\) and the CPU time to \(0.48s\) making it almost
\(100\) times faster. This normalization alternative increases the number
of queries labeled differently by the MPC to \(5\), however, it yields
\(266\) errors in total. For this specific dataset, the pre-computed
normalization used happens to gives better results than the original
baseline.

\paragraph{Cosine on MNIST.}
We also implement our solution using the MNIST dataset \cite{mnistdata}
(an image classification dataset where the task is to recognize digits
between 0 and 9 in the image),
this time with the Cosine algorithm.
We consider 15 sources, each training a logistic regression model for
MNIST
on a subset of the considered dataset. The client chooses \(120\) binary
queries to be answered by each source. To apply the MPC solution, the
sources secret-share the answers on two servers and evaluate using MPC
the Cosine truth-finding algorithm. At the end of the evaluation, the
results are reconstructed by the client. The results include the truth
value for each query (the binary label) and a trustworthiness factor for
each of the 15 sources. To evaluate MPC's impact, we compare the results
obtained to a base model. The base model corresponds to the Cosine
algorithm implemented without MPC on the same answers.
Fig.~\ref{fig:cosine-mnist} shows the difference
between predictions from the base model and from the
MPC evaluation.

The MPC evaluation contains errors compared to the base model,
mostly below \(10^{-3}\). To evaluate the impact of the errors
induced by MPC, we look at label prediction. The MPC method labels all
the questions exactly the same way as the baseline method, so both
methods made the same number of errors which is \(12\). On average, the
execution of each iteration took 0.47~s wall-clock time, or 0.36~s
CPU time. The MPC model can be up to \(4\,000\) times slower than the base model.

\begin{figure}[t]
    \centering
    \includegraphics[width=.9\linewidth]{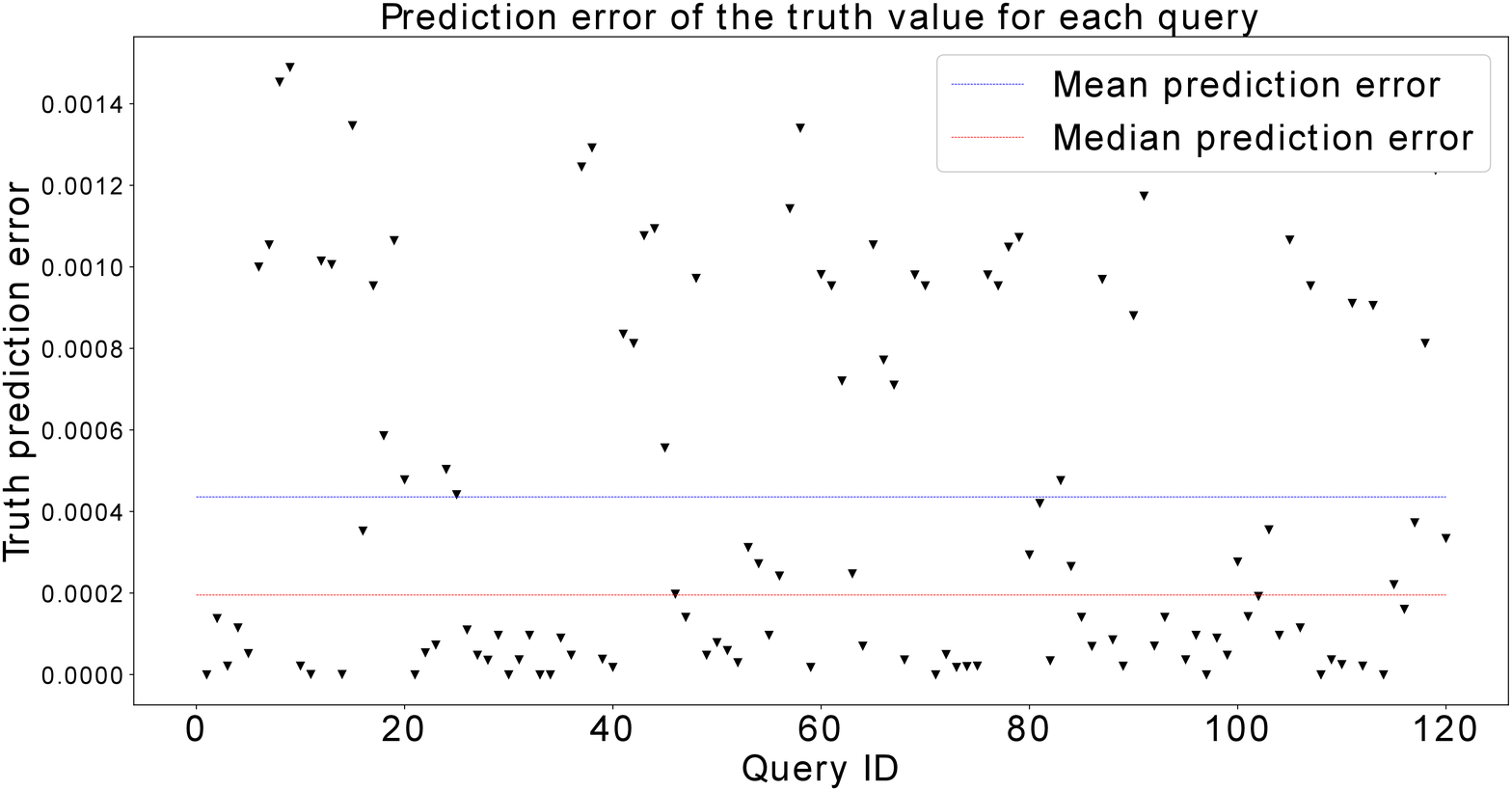}\\[1em]
    \includegraphics[width=.9\linewidth]{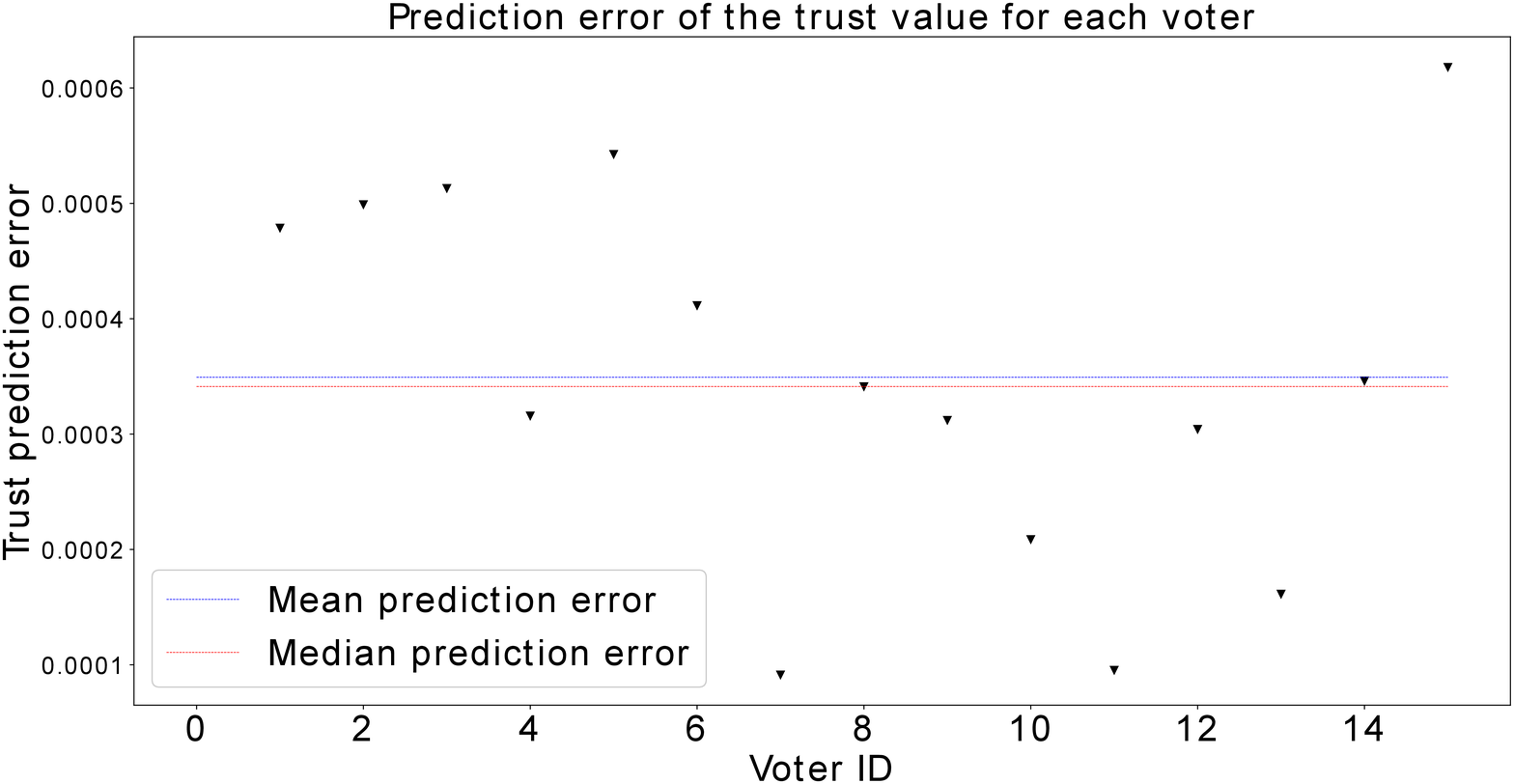}
    \caption{Prediction errors between secure multi-party computation and the base model results with Cosine on MNIST dataset.}
    \label{fig:cosine-mnist}
\end{figure}

If we apply the modifications for Cosine presented in Sec.~\ref{sec:alg2} the outputs will be very different of course because of
the aforementioned reasons. The wall-clock time of each iteration is barely
reduced to 0.44~s and the CPU time to 0.33~s. If there were more
sources, the time difference would have been more significant. This
alternative increases the number of queries labeled differently by the
MPC to \(2\), however, the number of errors is the same: \(12\).

\section{Further Discussion and Conclusion}

In this paper, we devised, presented, and evaluated the performance of MPC protocols for
truth-finding algorithms  corroborating information from disagreeing
views while preserving the confidentiality of the data in the sources. This solution is very
helpful to complete missing, uncertain, or rare data that is confidential
or sensitive, such as financial and medical data (or scientific data in
general). The MPC protocols we have proposed are very versatile and can be used to implement other algorithms securely, in particular our secure equality test alternative based on a simple polynomial evaluation.

The solution proposed can be further improved  by using MPC to protect the client's data and by using differential privacy techniques
\cite{dworkDP} to protect sources' privacy. Several works have demonstrated the possibility to combine MPC and differential privacy~\cite{idash-winner-article-concurrant, mypaper-icdis}. Indeed this would help further protect the models from
inversion attacks. Another application of
the model we propose would be combing MPC with regular voting and
distributed noise generation techniques
\cite{dp-for-mpc-thissen-master-thesis} to build a version of PATE
(private aggregation of teacher ensembles) \cite{pate} that keeps the
teacher's data private. In addition, using a truth-finding algorithm like
3-Estimates instead of regular voting for PATE might yield better
labeling of incomplete data. A research direction would be evaluating the
privacy, security, efficiency, and accuracy of different combinations of
tools like MPC, differential privacy, and truth-finding algorithms.

\subsubsection{Acknowledgments}
This research is part of the program DesCartes and is supported by the National
Research Foundation, Prime Minister’s Office, Singapore under its Campus for
Research Excellence and Technological Enterprise (CREATE) program.

\bibliographystyle{splncs04}
\bibliography{ref}

\end{document}